\documentclass[journal,twoside]{IEEEtran}
\usepackage{graphicx}
\usepackage{epstopdf}
\usepackage{placeins}
\usepackage{array}
\DeclareGraphicsExtensions{.eps,.pdf} 
\graphicspath{ {figures/} }
\usepackage{cite}
\usepackage{balance}
\usepackage{amssymb,amsmath}
\usepackage{subfigure}
\usepackage[ruled,vlined]{algorithm2e}
\usepackage{amssymb,amsmath,mathtools}
\usepackage{amssymb,amsmath}
\usepackage{algorithmic}
\usepackage{color}
\usepackage{multirow}
\makeatletter
\newcommand\restartchapters{\par
  \setcounter{chapter}{0}%
  \setcounter{section}{0}%
  \gdef\@chapapp{\chaptername}%
  \gdef\thechapter{\@arabic\c@chapter}}
\makeatother

\newcommand{\Pro}{{\mathtt{Prob}}}
\newtheorem{lemma}{\it \underline{Lemma}}

\newcommand{\F}{{\mathrm{F}}}

\newcommand{\tr}{{\mathrm{Tr}}}

\newcommand{\maxi}{{\mathtt{maximize}}}
\newcommand{\mini}{{\mathtt{minimize}}}

\newcommand{\st}{{\mathtt{s.t.}}}

\renewcommand{\algorithmicrequire}{\textbf{Input:}}

\newcommand{\ds}{\displaystyle}

\makeatletter
\g@addto@macro\normalsize{%
 \setlength\abovedisplayskip{2pt}
 \setlength\belowdisplayskip{2pt}
 \setlength\abovedisplayshortskip{2pt}
 \setlength\belowdisplayshortskip{2pt}
}
\makeatother
\begin{document}
\bstctlcite{IEEEexample:BSTcontrol}
\title{On the Design of Secure Full-Duplex Multiuser Systems under User Grouping Method}
\author{
\IEEEauthorblockN{Van-Dinh Nguyen$^{\dag}$, Hieu V. Nguyen$^{\dag}$,   Octavia A. Dobre$^{\ddag}$, and Oh-Soon Shin$^{\dag}$} \\
\IEEEauthorblockA{$^{\dag}$School of Electronic Engineering $\&$ Department of ICMC Convergence Technology, Soongsil University, Korea \\
                  	$^{\ddag}$Faculty of Engineering and Applied Science, Memorial University, St. John's, NL, Canada\\
										\vspace{-0.9cm}
				}
}
\maketitle
\thispagestyle{empty} 
\pagestyle{empty}
\begin{abstract}
Consider a full-duplex (FD) multiuser system where an FD base station (BS) is designed to concurrently serve both downlink   and uplink users in the presence of half-duplex eavesdroppers (Eves).  The target problem is to maximize the minimum secrecy rate (SR) among all legitimate users. A novel user grouping-based fractional time allocation is proposed as an alternative solution, where information signals at the FD-BS are accompanied by  artificial noise  to degrade the Eves'  channels. The SR problem has a highly non-concave and non-smooth objective, subject to non-convex constraints due to coupling between the optimization variables. Nevertheless, we develop a  path-following low-complexity algorithm, which involves only a simple convex program of moderate dimensions at each iteration.   Numerical results demonstrate the merit of the proposed approach compared to existing well-known ones, i.e., conventional FD and FD non-orthogonal multiple access.
\end{abstract}

\vspace{-0.5cm}
\section{Introduction} \label{Introduction}

By enabling simultaneous transmission and reception on the same channel, full duplex (FD) radio, which  has the potential of doubling the spectral efficiency compared to its half-duplex (HD) counterpart, has arisen as a promising technology for  5G wireless networks \cite{ZhangCM15,YadavAcess17}. The major challenge in designing an FD radio is to suppress the self-interference (SI) caused by the signal leakage from the downlink (DL) transmission to the uplink (UL) reception on the same device to a   suitable level, such as a few dB above the background noise. {\color{black}Fortunately, recent advances in hardware design have allowed the FD radio  to be  implemented at a reasonable cost while canceling a major part of the SI through analog circuits and digital signal processing \cite{DUPLO}.} 

Wireless networks have a very wide range of applications, and an unprecedented amount of personal information is transmitted over wireless channels. Consequently, wireless network security is a crucial issue  due to the unalterable open nature of the wireless medium. Physical-layer (PHY-layer) security can potentially provide information security at the PHY-layer by taking advantage of the characteristics  of the wireless medium. An effective means to deliver PHY-layer security is to adopt artificial noise (AN) to degrade the decoding capability of the eavesdropper (Eve) \cite{ChenCST16,Nguyen:TIFS:16}. Notably, with FD radio, we can exploit AN even more effectively \cite{ChenCST16}. With the FD radio at a base station (BS),   communication secrecy can be achieved for both UL and DL transmissions. In \cite{ZhuTSP14}, joint information  and AN beamforming at the FD-BS was investigated to guarantee the security of a single-antenna UL user  and DL user. However, this work assumed that there is neither SI nor co-channel interference (CCI) caused by an UL user's signal to a DL user, which  is highly idealistic. Therefore, an extension was proposed in \cite{ZhuTWC16} by considering both SI and CCI. The work in \cite{SunTWC16} analyzed a trade-off between DL and UL transmit power in FD systems to secure
multiple DL  and UL users. However, in practice, the harmful effect  of SI cannot be neglected if it is not properly controlled, and is proportional to the DL transmission power. Additionally, the CCI may become strong  whenever an UL user is located near  DL users. These shortcomings limit the performance of FD systems \cite{ZhuTSP14,ZhuTWC16,SunTWC16}.

In this paper, we propose a new transmission design to further resolve the practical restrictions mentioned above. Specifically, the near DL users and far UL users are served in a fraction of the time block, and then  FD-BS employs the remaining fractional time  to serve  near UL users and  far DL users. It is worth noting that the effects of SI, CCI and multiuser interference (MUI) are clearly reduced.  On the other hand,  FD-BS can  effectively perform  transmit beamforming even if the number of DL users exceeds the number of transmit antennas because the number of users that are served at the same time is effectively reduced. There are multiple-antenna eavesdroppers  that overhear the information signals from both DL and UL channels. We are concerned with the problem of jointly optimizing linear precoders/beamformers at the FD-BS and allocating the  UL transmit power,  as well as the fractional time (FT) to maximize the minimum secrecy rate (SR) among all users subject to power constraints. In general, such a design problem involves optimization of highly non-concave and non-smooth objective functions subject to non-convex constraints, for which the optimal solution is difficult to find.  The main contributions of the paper are summarized as follows:
\begin{enumerate}
 \item We propose a new transmission model for FD security to simultaneously optimize  both DL and UL information  privacy by exploring user grouping-based fractional
time model; this helps manage the  network interference more effectively than aiming to focus the interference at  Eves.
   \item We propose a path-following computational procedure to maximize the minimum SR by developing a new inner approximation of the original non-convex problem. The convex program solved at each iteration is of moderate dimension,  and thus is computationally efficient.
\item Numerical results show that the proposed FD scheme provides a substantial improvement of the SR performance over the conventional FD and FD non-orthogonal multiple access. 
\end{enumerate}

\emph{Notation}:  $\mathbf{X}^{H}$, $\mathbf{X}^{T}$ and $\tr(\mathbf{X})$  are the Hermitian transpose, normal transpose and trace of a matrix $\mathbf{X}$, respectively. $\|\cdot\|_{\text{F}}$, $\|\cdot\|$ and $|\cdot|$ denote the Frobenius matrix norm,  Euclidean norm of a vector, and  absolute value of a complex scalar, respectively.      $\Re\{\cdot\}$ represents the real part of the argument.   

\vspace{-0.2cm}
\section{System Model and  Problem Formulation} \label{System Model}
%
\subsection{Signal Processing Model}
\begin{figure}[t]
\centering
\includegraphics[width=0.48\textwidth,trim={-0cm 0.0cm 0cm 0cm}]{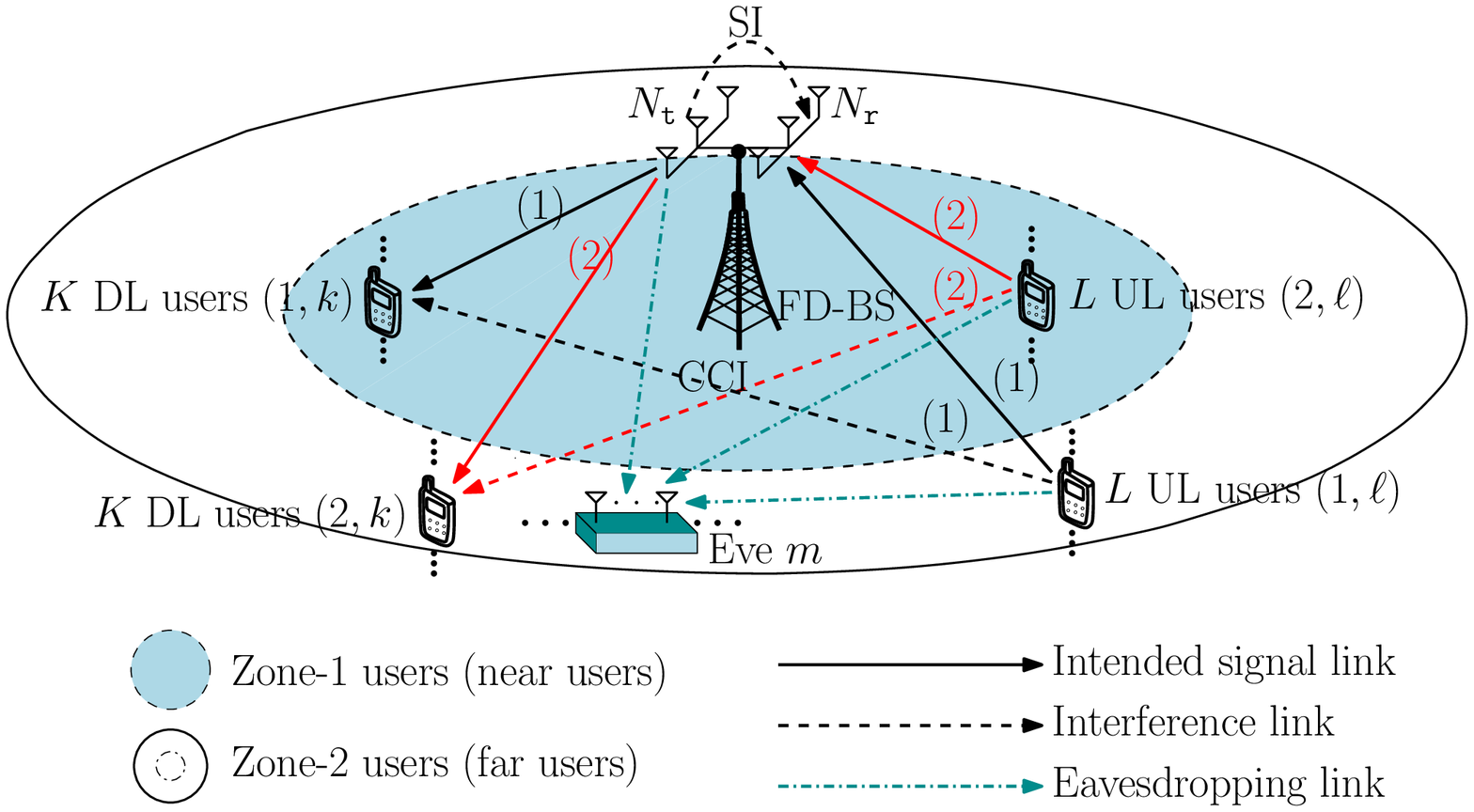}
\caption{A multiuser system model with an FD-BS serving $2K$ DL users and $2L$ UL users in the presence of $M$ Eves.}
\label{fig:SM:1}
\end{figure}
%

Consider a multiuser communication system illustrated in Fig.~\ref{fig:SM:1}, where the FD-BS is equipped with $N_{\mathtt{t}}$ transmit antennas and $N_{\mathtt{r}}$  receive antennas to simultaneously serve $2K$ DL users and $2L$ UL users over the same radio frequency band. Each legitimate user  is equipped with a single antenna to ensure low hardware complexity. The communications of both DL and UL are overheard by $M$ non-colluding  Eves, where the $m$-th Eve has $N_{e,m}$ antennas. Herein, we use a natural and efficient divisions of the coverage area \cite{Nguyen:JSAC:17} by dividing users  into two zones. To lighten the notation, we assume that
there are $K$ DL users and $L$ UL users located in a zone nearer the FD-BS (referred to as zone-1 of near users), and $K$ DL users and $L$ UL users are located in a zone farther from the FD-BS (called  zone-2 of far users). 

In this paper, we split each communication time block, denoted by $T$, into two sub-time blocks orthogonally. As previously mentioned, in order to mitigate the harmful effects of  SI, CCI and MUI,   $K$ near DL users and $L$ far UL users are grouped into group-1, and $K$ far DL users and $L$ near UL users are grouped into   group-2.  During the first duration  $\tau T\; (0 < \tau < 1)$,  users in group-1 are served while  users in group-2 are served in the remaining duration  $(1-\tau)T$. Although each group still operates in the FD mode, the inter-group interference, i.e., interference across groups 1 and 2, can be eliminated through the FT allocation.  Without loss of generality, the communication time block $T$ is normalized to 1. Upon denoting $\mathcal{K}\triangleq\{1,2,\cdots,K\}$ and $\mathcal{L}\triangleq\{1,2,\cdots,L\}$,  the sets of DL  and UL users are $\mathcal{D}\triangleq \mathcal{I}\times\mathcal{K}$ and $\mathcal{U}\triangleq \mathcal{I}\times\mathcal{L}$ for $\mathcal{I}\triangleq\{1,2\}$, respectively. Thus, the $k$-th DL user and the $\ell$-th UL user in the $i$-th group are referred to as DL user ($i,k$) and UL user ($i,\ell$), respectively.

\textit{1) Received Signal Model at the FD-BS and DL Users:}
 We consider that the FD-BS uses a transmit beamfomer $\mathbf{w}_{i,k}\in\mathbb{C}^{N_{\mathtt{t}}\times 1}$ to transfer the information bearing signal $x_{i,k}$, with $\mathbb{E}\{|x_{i,k}|^2\}=1$,  to DL user $(i,k)$. 
The FD-BS also injects an AN to interfere with the reception of the Eves as:
$\mathbf{x}_i = \sum_{k=1}^K\mathbf{w}_{i,k}x_{i,k} + \mathbf{v}_i$
for DL users in group-$i$, where $\mathbf{v}_i\in\mathbb{C}^{N_{\mathtt{t}}\times 1}, i=\{1,2\}$ is the AN vector whose elements are zero-mean complex Gaussian random variables, i.e., $\mathbf{v}_i\sim\mathcal{CN}(\mathbf{0},\mathbf{V}_i\mathbf{V}_i^H)$ with $\mathbf{V}_i\in\mathbb{C}^{N_{\mathtt{t}}\times N_{\mathtt{t}}}$.   The received signal at DL user $(i,k)$  can be expressed as
\begin{IEEEeqnarray}{rCl}\label{eq:signalDL1k}
y_{i,k} = \mathbf{h}_{i,k}^H\mathbf{w}_{i,k}x_{i,k} + \sum\nolimits_{j=1,j\neq k}^K\mathbf{h}_{i,k}^H\mathbf{w}_{i,j}x_{i,j} \nonumber\\ +\; \mathbf{h}_{i,k}^H\mathbf{v}_{i}  + \sum\nolimits_{\ell=1}^Lf_{i, k, \ell}\rho_{i,\ell}\tilde{x}_{i,\ell} + n_{i,k}
\end{IEEEeqnarray}
where $\mathbf{h}_{i,k}\in\mathbb{C}^{N_{\mathtt{t}}\times 1}$ is the transmit channel vector from the FD-BS to DL user $(i,k)$. In \eqref{eq:signalDL1k}, the term $\sum_{\ell=1}^Lf_{i, k, \ell}\rho_{i,\ell}\tilde{x}_{i,\ell}$ represents the CCI from $L$ UL users to DL user $(i,k)$, where $f_{i, k, \ell}\in\mathbb{C}$, $\rho_{i,\ell}$ and $\tilde{x}_{i,\ell}$ with $\mathbb{E}\{|\tilde{x}_{i,\ell}|^2\}=1$ are the complex channel coefficient from UL user $(i,\ell)$ to DL user $(i,k)$, transmit power and message of UL user $(i,\ell)$, respectively.  $n_{i,k}\sim\mathcal{CN}(0,\sigma^2)$  denotes the additive white Gaussian noise (AWGN) at DL user $(i,k)$. By defining $\tau_1:=\tau$ and $\tau_2:=1-\tau$, the information rate decoded by DL user $(i,k)$  in nat/sec/Hz is given by \cite{Nguyen:TCOM:17} 
\begin{IEEEeqnarray}{rCl}\label{eq:RateDLUs}
C_{i,k}^{\mathtt{D}}(\mathbf{X}_i,\tau_i) &=& \tau_i\ln\Bigl(1 + \frac{|\mathbf{h}_{i,k}^H\mathbf{w}_{i,k}|^2}{\varphi_{i,k}(\mathbf{X}_i)}\Bigr)
\end{IEEEeqnarray}
where $\mathbf{X}_i\triangleq  \bigl\{\mathbf{w}_{i},\mathbf{V}_i,\boldsymbol{\rho}_i\bigl\},$ with $\mathbf{w}_{i}\triangleq\{\mathbf{w}_{i,k}\}_{k\in\mathcal{K}}$, $\boldsymbol{\rho}_i\triangleq\{\rho_{i,\ell}\}_{\ell\in\mathcal{L}}, i=1,2,$  and 
$\varphi_{i,k}(\mathbf{X}_i) \triangleq \sum_{j=1,j\neq k}^K|\mathbf{h}_{i,k}^H\mathbf{w}_{i,j}|^2 + \|\mathbf{h}_{i,k}^H\mathbf{V}_{i}\|^2  + \sum_{\ell=1}^L\rho_{i,\ell}^2|f_{i, k, \ell}|^2 + \sigma^2.\nonumber$
The received signal at  the FD-BS for reception of $L$ UL users in the $i$-th group can be expressed as
\begin{IEEEeqnarray}{rCl}\label{eq:signalUL2p}
\mathbf{y}_{i,bs} =  \sum\nolimits_{\ell=1}^L\rho_{i,\ell}\mathbf{g}_{i,\ell}\tilde{x}_{i,\ell}  + \sqrt{\sigma_{\mathtt{SI}}}\sum\nolimits_{k=1}^K\mathbf{G}_{\mathtt{SI}}^H\mathbf{w}_{i,k}x_{i,k} \nonumber\\ +\; \sqrt{\sigma_{\mathtt{SI}}}\mathbf{G}_{\mathtt{SI}}^H\mathbf{v}_{i}  + \mathbf{n}_{i,bs}\qquad
\end{IEEEeqnarray}
where $\mathbf{g}_{i,\ell}\in\mathbb{C}^{N_{\mathtt{r}}\times 1}$ is the receive channel vector from UL user $(i,\ell)$ to the FD-BS. The term $\sqrt{\sigma_{\mathtt{SI}}}\sum_{k=1}^K\mathbf{G}_{\mathtt{SI}}^H\mathbf{w}_{i,k}x_{i,k}$ in (\ref{eq:signalUL2p}) represents the residual SI  after  cancellation in analog and digital domains; $\mathbf{G}_{\mathtt{SI}}\in\mathbb{C}^{N_{t}\times N_{r}}$ denotes a fading loop channel which impairs the UL signal detection at the FD-BS due to a concurrent DL transmission  and $0\leq\sigma_{\mathtt{SI}}< 1$ is used to model the degree of residual SI. $\mathbf{n}_{i,bs}\sim\mathcal{CN}(\mathbf{0},\sigma^2\mathbf{I}_{N_{\mathtt{r}}})$ denotes the AWGN at  the FD-BS. We adopt the minimum mean square error and successive interference
cancellation (MMSE-SIC) decoder at  the FD-BS \cite{Tse:book:05}.  Hence, the information rate in decoding the UL user $(i,\ell)$'s  message is given by \cite{Nguyen:TCOM:17} 
\begin{IEEEeqnarray}{rCl}\label{eq:RateULUs}
C_{i,\ell}^{\mathtt{U}}(\mathbf{X}_i,\tau_i) &=& \tau_i\ln\Bigl(1 + \rho_{i,\ell}^2\mathbf{g}_{i,\ell}^H\boldsymbol{\Phi}_{i,\ell}(\mathbf{X}_i)^{-1}\mathbf{g}_{i,\ell}  \Bigr)
\end{IEEEeqnarray}
where 
$\boldsymbol{\Phi}_{i,\ell}(\mathbf{X}_i) \triangleq \sum_{j>\ell}^L\rho_{i,j}^2\mathbf{g}_{i,j}\mathbf{g}_{i,j}^H + \sigma_{\mathtt{SI}}\sum_{k=1}^K\mathbf{G}_{\mathtt{SI}}^H\mathbf{w}_{i,k}\mathbf{w}_{i,k}^H\mathbf{G}_{\mathtt{SI}}  + \sigma_{\mathtt{SI}}\mathbf{G}_{\mathtt{SI}}^H\mathbf{V}_{i}\mathbf{V}_{i}^H\mathbf{G}_{\mathtt{SI}}  + \sigma^2\mathbf{I}_{N_{\mathtt{r}}}.$

\textit{2) Received Signal Model at Eves:}
The information signals of group-$i$ leaked out to the $m$-th Eve during the FT $\tau_i$ can be expressed as
\begin{equation}\label{eq:Eavsignal}
\mathbf{y}_{i,m} = \mathbf{H}_m^H\Bigl(\sum_{k=1}^K\mathbf{w}_{i,k}x_{i,k}+ \mathbf{v}_i\Bigr)+ \sum_{\ell=1}^L\rho_{i,\ell}\mathbf{g}_{m,i,\ell}^H\tilde{x}_{i,\ell} + \mathbf{n}_{e,m}
\end{equation}
where $\mathbf{H}_m\in\mathbb{C}^{N_{\mathtt{t}}\times N_{e,m}}$ and $\mathbf{g}_{m,i,\ell}\in\mathbb{C}^{1\times N_{e,m}}$ are the wiretap channel matrix and vector from the FD-BS and UL user ($i,\ell$) to the $m$-th Eve, respectively. $\mathbf{n}_{e,m}\sim\mathcal{CN}(\mathbf{0},\sigma^2\mathbf{I}_{N_{e,m}})$ denotes the AWGN at the $m$-th Eve. The  information  rates at the $m$-th Eve, corresponding to the signal targeted for DL user $(i,k)$ and UL user $(i,\ell)$, are given by
\begin{IEEEeqnarray}{rCl}\label{eq:RateEvam}
C_{m,i,k}^{\mathtt{ED}}(\mathbf{X}_i,\tau_i) &=& \tau_i\ln\bigl(1 + \|\mathbf{H}_{m}^H\mathbf{w}_{i,k}\|^2/\psi_{m,i,k}(\mathbf{X}_i)\bigr),\quad\IEEEyessubnumber\\
C_{m,i,\ell}^{\mathtt{EU}}(\mathbf{X}_i,\tau_i) &=& \tau_i\ln\bigl(1 + \rho_{i,\ell}^2\|\mathbf{g}_{m,i,\ell}^H\|^2/\chi_{m,i,\ell}(\mathbf{X}_i)\bigr)\IEEEyessubnumber
\end{IEEEeqnarray}
respectively, where
\begin{IEEEeqnarray}{rCl}
\psi_{m,i,k}(\mathbf{X}_i)&\triangleq& \sum\nolimits_{j=1,j\neq k}^K\|\mathbf{H}_{m}^H\mathbf{w}_{i,j}\|^2  + \|\mathbf{H}_{m}^H\mathbf{V}_{i}\|_\F^2 \nonumber\\ &+& \sum\nolimits_{\ell=1}^L\rho_{i,\ell}^2\|\mathbf{g}_{m,i,\ell}^H\|^2 + N_{e,m}\sigma^2,\nonumber\\
\chi_{m,i,\ell}(\mathbf{X}_i) &\triangleq& \sum\nolimits_{k=1}^K\|\mathbf{H}_{m}^H\mathbf{w}_{i,k}\|^2 + \|\mathbf{H}_{m}^H\mathbf{V}_{i}\|_\F^2 \nonumber\\ &+& \sum\nolimits_{j=1,j\neq \ell}^L \rho_{i,j}^2\|\mathbf{g}_{m,i,j}^H\|^2 + N_{e,m}\sigma^2.\nonumber
\end{IEEEeqnarray}

\vspace{-0.346cm}
\subsection{Optimization Problem Formulation} \label{OptimizationProblem}
We aim to jointly optimize  the transmit information vectors and AN matrices ($\mathbf{X} \triangleq \{\mathbf{X}_1, \mathbf{X}_2\}$), along with  the FT  ($\boldsymbol{\tau}\triangleq\{\tau_1,\tau_2\}$) to maximize the minimum (max-min) SR. The optimization problem  can be mathematically formulated as
\begin{IEEEeqnarray}{rCl}\label{eq:OP1}
\underset{\mathbf{X},\boldsymbol{\tau}}{\maxi}&&\;\underset{\substack{(i,k)\in\mathcal{D}, (i,\ell)\in\mathcal{U}}}{\mini}\;\left\{R_{i,k}^{\mathtt{D}}(\mathbf{X}_i,\tau_i),R_{i,\ell}^{\mathtt{U}}(\mathbf{X}_i,\tau_i)\right\} \IEEEyessubnumber\label{eq:OP1:a}\\
\st\ &&  \sum\nolimits_{i=1}^2\tau_i\Bigl(\sum\nolimits_{k=1}^K\|\mathbf{w}_{i,k}\|^2 + \|\mathbf{V}_i\|_\F^2\Bigr) \leq  P_{bs}^{\max},   \IEEEyessubnumber\label{eq:OP1:b}\qquad \\
&& \tau_i \rho_{i,\ell}^2 \leq P_{i,\ell}^{\max},\;\forall (i,\ell)\in\mathcal{U},  \IEEEyessubnumber\label{eq:OP1:d}\\
&& \rho_{i,\ell} \geq 0, \,\forall (i,\ell)\in\mathcal{U}, \IEEEyessubnumber\label{eq:OP1:e}\\
&& \tau_1>0, \tau_2>0, \tau_1 + \tau_2 \leq 1\IEEEyessubnumber\label{eq:OP1:f}
\end{IEEEeqnarray}
where $\mathcal{M}\triangleq\{1,2,\cdots,M\}$ and
\begin{IEEEeqnarray}{rCl}\label{eq:Secrecyrates}
R_{i,k}^{\mathtt{D}}(\mathbf{X}_i,\tau_i) &\triangleq& \bigl[C_{i,k}^{\mathtt{D}}(\mathbf{X}_i,\tau_i) - \underset{m\in\mathcal{M}}{\max}\,C_{m,i,k}^{\mathtt{ED}}(\mathbf{X}_i,\tau_i)\bigr]^+,\quad\  \IEEEyessubnumber\\
R_{i,\ell}^{\mathtt{U}}(\mathbf{X}_i,\tau_i) &\triangleq& \bigl[C_{i,\ell}^{\mathtt{U}}(\mathbf{X}_i,\tau_i) - \underset{m\in\mathcal{M}}{\max}\,C_{m,i,\ell}^{\mathtt{EU}}(\mathbf{X}_i,\tau_i)\bigl]^+\IEEEyessubnumber
\end{IEEEeqnarray}
with $[x]^+\triangleq\max\{0,x\}$.
Constraint \eqref{eq:OP1:b} merely means that the total transmit power at the FD-BS does not exceed the  power budget, $P_{bs}^{\max}$ \cite{Nguyen:TCOM:17,Nguyen:Access:17}, while  constraints in \eqref{eq:OP1:d} are  individual power budgets at the UL user ($i,\ell$), $P_{i,\ell}^{\max}$. 
 
\vspace{-0.18cm}
\section{Proposed Optimal Solution}\label{sec:knownCSI}
Similar to prior work  dealing with the resource allocation in FD systems,  perfect instantaneous channel state information (CSI) of the legitimate users is assumed to be available at the transmitters \cite{Nguyen:TCOM:17,YadavAcess17}. On the other hand, we consider that Eves are always passive and do not transmit. In this case, we assume that only the statistics of CSI  for Eves (i.e., the first- and second-order statistics) are available at the transmitter \cite{AkgunTCOM17}, i.e.,
\begin{IEEEeqnarray}{rCl}
\bar{\mathbf{H}}_m = \mathbb{E}\left\{\mathbf{H}_m\mathbf{H}_m^H\right\}\ \text{and}\
\bar{g}_{m,i,\ell} =  \mathbb{E}\left\{\mathbf{g}_{m,i,\ell}\mathbf{g}_{m,i,\ell}^H\right\}.
\label{eq:CDI}\end{IEEEeqnarray}
\vspace{-0.6cm}
\subsection{Equivalent Transformations for \eqref{eq:OP1}}
 We first introduce the new variables $\eta$ and $\boldsymbol{\Gamma}\triangleq\bigr\{\Gamma_{i,k}^{\mathtt{D}},\Gamma_{i,\ell}^{\mathtt{U}}\bigr\}_{i\in\mathcal{I},k\in\mathcal{K},\ell\in\mathcal{L}}$ to equivalently re-write \eqref{eq:OP1} as: 
\begin{IEEEeqnarray}{rCl}\label{eq:OP3}
&&\underset{\mathbf{X},\boldsymbol{\tau},\eta,\boldsymbol{\Gamma}}{\maxi}\quad \eta \IEEEyessubnumber\label{eq:OP3:a}\\
&&\st\; \eqref{eq:OP1:b},  \eqref{eq:OP1:d}, \eqref{eq:OP1:e}, \eqref{eq:OP1:f}\IEEEyessubnumber\label{eq:OP3:j},\\
&&\qquad C_{i,k}^{\mathtt{D}}(\mathbf{X}_i,\tau_i) - \Gamma_{i,k}^{\mathtt{D}}\geq \eta,\ \forall (i,k)\in\mathcal{D},   \IEEEyessubnumber\label{eq:OP3:b}\\
&&\qquad C_{m,i,k}^{\mathtt{ED}}(\mathbf{X}_i,\tau_i) \leq \Gamma_{i,k}^{\mathtt{D}}, \forall m\in\mathcal{M}, (i,k)\in\mathcal{D},\IEEEyessubnumber\label{eq:OP3:c}\\
&&\qquad  C_{i,\ell}^{\mathtt{U}}(\mathbf{X}_i,\tau_i) - \Gamma_{i,\ell}^{\mathtt{U}}\geq \eta,\ \forall(i,\ell)\in\mathcal{U},   \IEEEyessubnumber\label{eq:OP3:f}\\
&&\qquad C_{m,i,\ell}^{\mathtt{EU}}(\mathbf{X}_i,\tau_i) \leq \Gamma_{i,\ell}^{\mathtt{U}},\ \forall m\in\mathcal{M}, (i,\ell)\in\mathcal{U}.\IEEEyessubnumber\label{eq:OP3:g}
\end{IEEEeqnarray}
Problem \eqref{eq:OP3} still remains intractable. To solve it, we make the variable change:
\begin{equation}\label{eq:changevariables:a}
  \tau_1 = 1/\alpha_1\ \text{and}\  \tau_2 = 1/\alpha_2
\end{equation}
which implies the following convex constraint
\begin{equation}\label{eq:changevariables:b}
  1/\alpha_1 +  1/\alpha_2 \leq 1, \forall \alpha_i > 1, i\in\mathcal{I}
\end{equation}
where $\boldsymbol{\alpha}\triangleq\{\alpha_1,\alpha_2\}$ are new variables. Using \eqref{eq:changevariables:a},  constraints \eqref{eq:OP3:b} and \eqref{eq:OP3:f} become
\begin{IEEEeqnarray}{rCl}\label{eq:LegitimateUEs}
C_{i,k}^{\mathtt{D}}(\mathbf{X}_i,\alpha_i) &\geq& \eta + \Gamma_{i,k}^{\mathtt{D}},\ \forall (i,k)\in\mathcal{D},\IEEEyessubnumber\label{eq:changeOP3:a}\qquad\\
C_{i,\ell}^{\mathtt{U}}(\mathbf{X}_i,\alpha_i) &\geq& \eta + \Gamma_{i,\ell}^{\mathtt{U}},\ \forall(i,\ell)\in\mathcal{U}.\IEEEyessubnumber\label{eq:changeOP3:c}
\end{IEEEeqnarray}
For a safe design as in \cite{Nguyen:TIFS:16},  we consider the replacement of  constraints \eqref{eq:OP3:c} and \eqref{eq:OP3:g} by their minimum outage requirement 
\begin{IEEEeqnarray}{rCl}\label{eq:RateEvamChange}
\Pro\bigl(\underset{m\in\mathcal{M}}{\max} C_{m,i,k}^{\mathtt{ED}}(\mathbf{X}_i,\alpha_i)\leq \Gamma_{i,k}^{\mathtt{D}}\bigl) &\geq& \epsilon_{i,k},   \forall(i,k)\in\mathcal{D},\IEEEyessubnumber\label{eq:RateEvamChange:a}\\
\Pro\bigl(\underset{m\in\mathcal{M}}{\max}C_{m,i,\ell}^{\mathtt{EU}}(\mathbf{X}_i,\alpha_i)\leq \Gamma_{i,\ell}^{\mathtt{U}}\bigl) &\geq& \tilde{\epsilon}_{i,\ell}, \forall(i,\ell)\in\mathcal{U}\qquad\quad \IEEEyessubnumber\label{eq:RateEvamChange:c}
\end{IEEEeqnarray}
where $\epsilon_{i,k}$ and $\tilde{\epsilon}_{i,\ell}$ are  given values. From \eqref{eq:LegitimateUEs} and \eqref{eq:RateEvamChange}, and by substituting \eqref{eq:changevariables:a} and \eqref{eq:changevariables:b} to \eqref{eq:OP1:b}-\eqref{eq:OP1:d}, the optimization problem \eqref{eq:OP3} is equivalently re-expressed as
\begin{IEEEeqnarray}{rCl}\label{eq:OP4}
\underset{\mathbf{X},\eta,\boldsymbol{\Gamma},\boldsymbol{\alpha}}{\maxi}&&\quad \eta \IEEEyessubnumber\label{eq:OP4:a}\\
\st\;&&  \eqref{eq:OP1:e}, \eqref{eq:changevariables:b}, \eqref{eq:LegitimateUEs}, \eqref{eq:RateEvamChange},  \IEEEyessubnumber\label{eq:OP4:b}\\
&& \bigl(1-1/\alpha_2\bigr)\Bigl(\sum\nolimits_{k=1}^K\|\mathbf{w}_{1,k}\|^2 + \|\mathbf{V}_1\|_{\mathrm{F}}^2\Bigr) \nonumber\\
&& +\; \frac{1}{\alpha_2}\Bigl(\sum\nolimits_{k=1}^{K}\|\mathbf{w}_{2,k}\|^2 + \|\mathbf{V}_2\|_{\mathrm{F}}^2\Bigr) \leq  P_{bs}^{\max},   \IEEEyessubnumber\label{eq:OP4:c}\qquad\\
&& \bigl(1-1/\alpha_2\bigl)\rho_{1,\ell}^2 \leq P_{1,\ell}^{\max},\;\forall \ell\in\mathcal{L},  \IEEEyessubnumber\label{eq:OP4:d}\\
&& \rho_{2,\ell}^2/\alpha_2 \leq P_{2,\ell}^{\max},\;\forall \ell\in\mathcal{L}.  \IEEEyessubnumber\label{eq:OP4:e}
\end{IEEEeqnarray}

\vspace{-0.25cm}
\subsection{Proposed Convex Approximation-Based Iterations}
Before proceeding further, we note that except for \eqref{eq:OP1:e}, \eqref{eq:changevariables:b} and \eqref{eq:OP4:e}, the constraints are  non-convex.
{\color{black}The proposed method is mainly based on an inner approximation framework \cite{Marks:78}  to handle the non-convex parts.}

\textit{Convex Approximation of  Constraints \eqref{eq:LegitimateUEs}:} 
We first introduce the following inequality at a feasible point $(\gamma^{(\kappa)},t^{(\kappa)})$:
\begin{eqnarray}\label{eq:ineupper}
\ds\zeta(\gamma,t)\triangleq\frac{\ln(1+\gamma)}{t}&\geq&\ds \mathtt{A}^{(\kappa)} - \mathtt{B}^{(\kappa)}\frac{1}{\gamma} - \mathtt{C}^{(\kappa)}t\label{inq1}
\end{eqnarray}
for all $\gamma > 0, \gamma^{(\kappa)} > 0, t>0, t^{(\kappa)}>0$, where $\mathtt{A}^{(\kappa)}\triangleq 2\zeta(\gamma^{(\kappa)},t^{(\kappa)}) + \frac{\gamma^{(\kappa)}}{t^{(\kappa)}(\gamma^{(\kappa)}+1)},$ $
\mathtt{B}^{(\kappa)}\triangleq\frac{(\gamma^{(\kappa)})^2}{t^{(\kappa)}(\gamma^{(\kappa)}+1)},$ and 
$\mathtt{C}^{(\kappa)}\triangleq\frac{\zeta(\gamma^{(\kappa)},t^{(\kappa)})}{t^{(\kappa)}}.$ The proof of \eqref{inq1} is omitted due to the space limitation.   In the spirit of \cite{WES06}, for $\bar{\mathbf{w}}_{i,k}=e^{-j\mathtt{arg}(\mathbf{h}_{i,k}^H\mathbf{w}_{i,k})}\mathbf{w}_{i,k}$ with $j=\sqrt{-1}$, it follows that $|\mathbf{h}_{i,k}^H\mathbf{w}_{i,k}|=\mathbf{h}_{i,k}^H\bar{\mathbf{w}}_{i,k}=\Re\{\mathbf{h}_{i,k}^H\bar{\mathbf{w}}_{i,k}\}\geq 0$ and $|\mathbf{h}_{i',k'}^H\mathbf{w}_{i,k}|=|\mathbf{h}_{i',k'}^H\bar{\mathbf{w}}_{i,k}|$ for all $(i',k')\neq(i,k)$. Thus, $\gamma_{i,k}^{\mathtt{D}}(\mathbf{X}_i)\triangleq|\mathbf{h}_{i,k}^H\mathbf{w}_{i,k}|^2/\varphi_{i,k}(\mathbf{X}_i)$ can be equivalently replaced by
$\gamma_{i,k}^{\mathtt{D}}(\mathbf{X}_i)=\bigr(\Re\{\mathbf{h}_{i,k}^H\mathbf{w}_{i,k}\}\bigl)^2/\varphi_{i,k}(\mathbf{X}_i)$
with the condition 
\begin{IEEEeqnarray}{rCl}\label{eq:poscondi}
\Re\{\mathbf{h}_{i,k}^H\mathbf{w}_{i,k}\}\geq 0,\ \forall (i,k)\in\mathcal{D}. 
 \end{IEEEeqnarray}
By using \eqref{eq:ineupper},  $C_{i,k}^{\mathtt{D}}(\mathbf{X}_i,\alpha_i)$ in \eqref{eq:changeOP3:a} is lower bounded at  a feasible point $(\mathbf{X}_i^{(\kappa)},\alpha_i^{(\kappa)})$ found at the $(\kappa$-1)-th iteration by
\begin{IEEEeqnarray}{rCl}
  \frac{\ln\bigl(1+ \gamma_{i,k}^{\mathtt{D}}(\mathbf{X}_i)  \bigr)}{\alpha_i} &\geq& \mathtt{A}_{i,k}^{(\kappa)} - \mathtt{B}_{i,k}^{(\kappa)}\frac{\varphi_{i,k}(\mathbf{X}_i)}{\bigr(\Re\{\mathbf{h}_{i,k}^H\mathbf{w}_{i,k}\}\bigl)^2}- \mathtt{C}_{i,k}^{(\kappa)}\alpha_i \qquad
	\label{eq:Rate1kappro}
 \end{IEEEeqnarray}
where
$\mathtt{A}_{i,k}^{(\kappa)} \triangleq 2C_{i,k}^{\mathtt{D}}\bigl(\mathbf{X}_i^{(\kappa)},\alpha_i^{(\kappa)}\bigl) + \frac{\gamma_{i,k}^{\mathtt{D}}(\mathbf{X}_i^{(\kappa)})}{\alpha_i^{(\kappa)}\bigl(\gamma_{i,k}^{\mathtt{D}}(\mathbf{X}_i^{(\kappa)}) +1\bigr)},$ $
\mathtt{B}_{i,k}^{(\kappa)} \triangleq \frac{\bigl(\gamma_{i,k}^{\mathtt{D}}(\mathbf{X}_i^{(\kappa)})\bigr)^2}{\alpha_i^{(\kappa)}\bigl(\gamma_{i,k}^{\mathtt{D}}(\mathbf{X}_i^{(\kappa)})+1\bigr)},$ and $
\mathtt{C}_{i,k}^{(\kappa)}\triangleq \frac{C_{i,k}^{\mathtt{D}}\bigl(\mathbf{X}_i^{(\kappa)},\alpha_i^{(\kappa)}\bigl)}{\alpha_i^{(\kappa)}}.$
We make use of the inequality $\|\mathbf{x}\|^2 \geq 2\Re\{(\mathbf{x}^{(\kappa)})^H\mathbf{x}\} - \|\mathbf{x}^{(\kappa)}\|^2, \forall \mathbf{x}\in\mathbb{C}^N, \mathbf{x}^{(\kappa)}\in\mathbb{C}^N$ due to the convexity of the function $\|\mathbf{x}\|^2$ to further expose the hidden convexity of  \eqref{eq:Rate1kappro} as
\begin{IEEEeqnarray}{rCl}
  \frac{\ln\bigl(1+ \gamma_{i,k}^{\mathtt{D}}(\mathbf{X}_i)  \bigr)}{\alpha_i} &\geq& \mathtt{A}_{i,k}^{(\kappa)} - \mathtt{B}_{i,k}^{(\kappa)}\frac{\varphi_{i,k}(\mathbf{X}_i)}{\Psi_{i,k}^{(\kappa)}(\mathbf{w}_{i,k})}
	- \mathtt{C}_{i,k}^{(\kappa)}\alpha_i \nonumber\\
	&:=&C_{i,k}^{\mathtt{D},(\kappa)}(\mathbf{X}_i,\alpha_i)
	\label{eq:Rate1ConvexAppro}
 \end{IEEEeqnarray}
 over the trust region
\begin{equation}\label{eq:R1ktrust}
2\Re\{\mathbf{h}_{i,k}^H\mathbf{w}_{i,k}\}-\Re\{\mathbf{h}_{i,k}^H\mathbf{w}_{i,k}^{(\kappa)}\} > 0,\ \forall(i,k)\in\mathcal{D}
\end{equation}
where 
$\Psi_{i,k}^{(\kappa)}(\mathbf{w}_{i,k})\triangleq \Re\{\mathbf{h}_{i,k}^H\mathbf{w}_{i,k}^{(\kappa)}\}\bigr(2\Re\{\mathbf{h}_{i,k}^H\mathbf{w}_{i,k}\}-\Re\{\mathbf{h}_{i,k}^H\mathbf{w}_{i,k}^{(\kappa)}\}\bigl).$
Note that $C_{i,k}^{\mathtt{D},(\kappa)}(\mathbf{X}_i,\alpha_i)$ is a lower bounding concave function of $C_{i,k}^{\mathtt{D}}(\mathbf{X}_i,\alpha_i)$, which also satisfies
$C_{i,k}^{\mathtt{D},(\kappa)}\bigl(\mathbf{X}_i^{(\kappa)},\alpha_i^{(\kappa)}\bigl) = C_{i,k}^{\mathtt{D}}\bigl(\mathbf{X}_i^{(\kappa)},\alpha_i^{(\kappa)}\bigl).$
As a result,  \eqref{eq:changeOP3:a} can be iteratively replaced by the following convex constraint:
\begin{IEEEeqnarray}{rCl}\label{eq:R1kConvex}
C_{i,k}^{\mathtt{D},(\kappa)}(\mathbf{X}_i,\alpha_i) \geq \eta + \Gamma_{i,k}^{\mathtt{D}},\ (i,k)\in\mathcal{D}.
\end{IEEEeqnarray}

By defining  $\gamma_{i,\ell}^{\mathtt{U}}(\mathbf{X}_i)\triangleq \rho_{i,\ell}^2\mathbf{g}_{i,\ell}^H\boldsymbol{\Phi}_{i,\ell}(\mathbf{X}_i)^{-1}\mathbf{g}_{i,\ell}$, the left hand-side (LHS) of \eqref{eq:changeOP3:c} is lower bounded at the feasible point $\bigl(\mathbf{X}_i^{(\kappa)},\alpha_i^{(\kappa)}\bigr)$ as
\begin{IEEEeqnarray}{rCl}
   \frac{\ln\bigl(1 + \gamma_{i,\ell}^{\mathtt{U}}(\mathbf{X}_i)  \bigr)}{\alpha_i} &\geq& \tilde{\mathtt{A}}_{i,\ell}^{(\kappa)} + \tilde{\mathtt{B}}_{i,\ell}^{(\kappa)}\rho_{i,\ell} - \frac{\phi_{i,\ell}^{(\kappa)}\bigl(\mathbf{X}_i\bigr)}{\alpha_i^{(\kappa)}} - \tilde{\mathtt{C}}_{i,\ell}^{(\kappa)}\alpha_i \nonumber\\
	&:=& C_{i,\ell}^{\mathtt{U},(\kappa)}(\mathbf{X}_i,\alpha_i)
	\label{eq:Rate1qappro}
 \end{IEEEeqnarray}
where 
$\tilde{\mathtt{A}}_{i,\ell}^{(\kappa)} \triangleq  2C_{i,\ell}^{\mathtt{U}}\bigl(\mathbf{X}_i^{(\kappa)},\alpha_i^{(\kappa)}\bigr) - \gamma_{i,\ell}^{\mathtt{U}}(\mathbf{X}_i^{(\kappa)})/\alpha_i^{(\kappa)},
\tilde{\mathtt{B}}_{i,\ell}^{(\kappa)} \triangleq 2\gamma_{i,\ell}^{\mathtt{U}}(\mathbf{X}_i^{(\kappa)})/(\rho_{i,\ell}^{(\kappa)}\alpha_i^{(\kappa)}),\   \tilde{\mathtt{C}}_{i,\ell}^{(\kappa)} \triangleq   C_{i,\ell}^{\mathtt{U}}\bigl(\mathbf{X}_i^{(\kappa)},\alpha_i^{(\kappa)}\bigr)/\alpha_i^{(\kappa)},\                         
\phi_{i,\ell}^{(\kappa)}\bigl(\mathbf{X}_i\bigr) \triangleq \tr\bigl(\bigl(\rho_{i,\ell}^2\mathbf{g}_{i,\ell}\mathbf{g}_{i,\ell}^H + \boldsymbol{\Phi}_{i,\ell}(\mathbf{X}_i)\bigr)\boldsymbol{\Omega}_{i,\ell}^{(\kappa)} \bigr), $ and $ 
 \boldsymbol{\Omega}_{i,\ell}^{(\kappa)} \triangleq \boldsymbol{\Phi}_{i,\ell}\bigl(\mathbf{X}_i^{(\kappa)}\bigr)^{-1} - \boldsymbol{\Phi}_{i,\ell-1}\bigl(\mathbf{X}_i^{(\kappa)}\bigr)^{-1}\succeq\mathbf{0}.$ 
It follows from \eqref{eq:Rate1qappro} that $C_{i,\ell}^{\mathtt{U},(\kappa)}(\mathbf{X}_i,\alpha_i)$ is a concave  function, which agrees with $C_{i,\ell}^{\mathtt{U}}(\mathbf{X}_i,\alpha_i)$ at the feasible point $\bigl(\mathbf{X}_i^{(\kappa)},\alpha_i^{(\kappa)}\bigr)$ as
$C_{i,\ell}^{\mathtt{U},(\kappa)}\bigl(\mathbf{X}_i^{(\kappa)},\alpha_i^{(\kappa)}\bigr) = 
C_{i,\ell}^{\mathtt{U}}\bigl(\mathbf{X}_i^{(\kappa)},\alpha_i^{(\kappa)}\bigr).$
Thus, the constraint \eqref{eq:changeOP3:c} can be iteratively 
replaced by
\begin{equation}\label{eq:R1qConvexappro}
C_{i,\ell}^{\mathtt{U},(\kappa)}\bigl(\mathbf{X}_i,\alpha_i\bigr) \geq \eta + \Gamma_{i,\ell}^{\mathtt{U}},\ \forall(i,\ell)\in\mathcal{U}.
\end{equation}

\textit{Convex Approximation of  Constraints \eqref{eq:RateEvamChange}:}  For a given feasible point $x^{(\kappa)}$, the following inequality holds true:
\begin{equation}\label{ineupp}
\ln(1+x) \leq \mathtt{a}(x^{(\kappa)}) + \mathtt{b}(x^{(\kappa)})x,\;\forall x^{(\kappa)} \geq 0, x\geq 0
\end{equation}
where $\mathtt{a}(x^{(\kappa)}) \triangleq \ln(1+x^{(\kappa)}) - \frac{x^{(\kappa)}}{1+x^{(\kappa)}},\
 \mathtt{b}(x^{(\kappa)}) \triangleq  \frac{1}{1+x^{(\kappa)}},$ which is a result of the concavity of the function $\ln(1+x)$. For the concave function $\sqrt{yz}$, its convex upper bound is \cite{Beck:JGO:10}
\begin{eqnarray}
\sqrt{yz} &\leq& \frac{\sqrt{y^{(\kappa)}}}{2\sqrt{z^{(\kappa)}}}z + \frac{\sqrt{z^{(\kappa)}}}{2\sqrt{y^{(\kappa)}}}y \label{B3}
\end{eqnarray}
with $\forall y > 0, y^{(\kappa)} > 0, z > 0, z^{(\kappa)} > 0$. To evaluate \eqref{eq:RateEvamChange:a} and \eqref{eq:RateEvamChange:c}, we first introduce the following lemma.
\begin{lemma}\label{lemma:1}
Assuming all Eves have the same channel properties and  are independent,  \eqref{eq:RateEvamChange:a} and \eqref{eq:RateEvamChange:c}  are respectively converted  into the following constraints:
\begin{IEEEeqnarray}{rCl}
\frac{\mathbf{w}_{i,k}^H\bar{\mathbf{H}}_{m}\mathbf{w}_{i,k}}{e^{\alpha_i\Gamma_{i,k}^{\mathtt{D}}}-1}  \leq \bar{\psi}_{m,i,k}(\mathbf{X}_i) + \bigl(1-\epsilon_{i,k}^{1/M}\bigr)N_{e,m}\sigma^2 \label{eq:OP7:c1}\quad
 \end{IEEEeqnarray}
and 
\begin{IEEEeqnarray}{rCl}
\frac{\rho_{i,\ell}^2\bar{g}_{m,i,\ell} }{e^{\alpha_i\Gamma_{i,\ell}^{\mathtt{U}}}-1} \leq \bar{\chi}_{m,i,\ell}(\mathbf{X}_i) +\bigl(1-\tilde{\epsilon}_{i,\ell}^{1/M}\bigl) N_{e,m}\sigma^2\label{eq:OP7:d1}\quad
 \end{IEEEeqnarray}
 where 
$\bar{\psi}_{m,i,k}(\mathbf{X}_i) \triangleq\sum\nolimits_{j=1,j\neq k}^K\mathbf{w}_{i,j}^H\bar{\mathbf{H}}_{m}\mathbf{w}_{i,j} +\tr\bigl(\mathbf{V}_{i}^H\bar{\mathbf{H}}_{m}\mathbf{V}_{i}\bigr) + \sum\nolimits_{\ell=1}^L\rho_{i,\ell}^2\bar{g}_{m,i,\ell}$ and $
\bar{\chi}_{m,i,\ell}(\mathbf{X}_i) \triangleq \sum\nolimits_{k=1}^K\mathbf{w}_{i,k}^H\bar{\mathbf{H}}_{m}\mathbf{w}_{i,k} +\tr\bigl(\mathbf{V}_{i}^H\bar{\mathbf{H}}_{m}\mathbf{V}_{i}\bigr) + \sum\nolimits_{j=1,j\neq\ell}^L\rho_{i,j}^2\bar{g}_{m,i,j}.$
\end{lemma}
\begin{proof}
See the Appendix.
\end{proof}

We note that  constraint \eqref{eq:OP7:c1} is still non-convex but can be further shaped to take the equivalent  form:
\begin{IEEEeqnarray}{rCl}\label{eq:OP7:c3}
&&\mathbf{w}_{i,k}^H\bar{\mathbf{H}}_{m}\mathbf{w}_{i,k}/\beta_{i,k}^{\mathtt{D}}  \leq \bar{\psi}_{m,i,k}(\mathbf{X}_i) + \bigl(1-\epsilon_{i,k}^{1/M}\bigr)N_{e,m}\sigma^2, \IEEEyessubnumber\label{eq:OP7:c3a:nonconvex}\qquad\\
&& \beta_{i,k}^{\mathtt{D}} \leq e^{\alpha_i\Gamma_{i,k}^{\mathtt{D}}}-1 \Leftrightarrow \ln(1 + \beta_{i,k}^{\mathtt{D}})/\alpha_i \leq \Gamma_{i,k}^{\mathtt{D}} \IEEEyessubnumber\label{eq:OP7:c3b}\end{IEEEeqnarray}
where $\beta_{i,k}^{\mathtt{D}} > 0, \forall(i,k)\in\mathcal{D}$ are new variables. For \eqref{eq:OP7:c3a:nonconvex}, its LHS is a quadratic-over-affine function (which is convex) and the first term of the right hand-side (RHS) is a quadratic function. Then, we iteratively replace \eqref{eq:OP7:c3a:nonconvex} by
\begin{IEEEeqnarray}{rCl}\label{eq:OP7:c31}
\mathbf{w}_{i,k}^H\bar{\mathbf{H}}_{m}\mathbf{w}_{i,k}/\beta_{i,k}^{\mathtt{D}}  \leq \bar{\psi}_{m,i,k}^{(\kappa)}(\mathbf{X}_i) + \bigl(1-\epsilon_{i,k}^{1/M}\bigr)N_{e,m}\sigma^2 \qquad
\end{IEEEeqnarray}
 where $\bar{\psi}_{m,i,k}^{(\kappa)}(\mathbf{X}_i)$  $ \triangleq 2\bigl(\sum\nolimits_{j=1,j\neq k}^K\Re\{(\mathbf{w}_{i,j}^{(\kappa)})^H\bar{\mathbf{H}}_{m}\mathbf{w}_{i,j}\} +\Re\{\tr\bigl((\mathbf{V}_{i}^{(\kappa)})^H\bar{\mathbf{H}}_{m}\mathbf{V}_{i}\bigr)\}+\sum\nolimits_{\ell=1}^L\rho_{i,\ell}^{(\kappa)}\rho_{i,\ell}\bar{g}_{m,i,\ell}\bigl) - \bar{\psi}_{m,i,k}(\mathbf{X}_i^{(\kappa)})$ is the inner approximation of $\bar{\psi}_{m,i,k}(\mathbf{X}_i)$.   
 By using \eqref{ineupp},  \eqref{eq:OP7:c3b} holds that
\begin{IEEEeqnarray}{rCl}\label{eq:OP7:c3b1}
\mathtt{a}(\beta_{i,k}^{\mathtt{D},{(\kappa)}})/\alpha_i + \mathtt{b}(\beta_{i,k}^{\mathtt{D},{(\kappa)}})\beta_{i,k}^{\mathtt{D}}/\alpha_i \leq \Gamma_{i,k}^{\mathtt{D}}.
\end{IEEEeqnarray}
For $\mathcal{W}(\beta_{i,k}^{\mathtt{D}},\alpha_i)\triangleq\beta_{i,k}^{\mathtt{D}}/\alpha_i$, applying \eqref{B3} yields
\begin{IEEEeqnarray}{rCl}\label{eq:approW}
\mathcal{W}(\beta_{i,k}^{\mathtt{D}},\alpha_i) &\leq& \frac{1}{2}\Bigl(\frac{(\beta_{i,k}^{\mathtt{D}})^2}{\beta_{i,k}^{\mathtt{D},(\kappa)}}\frac{1}{\alpha_i^{(\kappa)}} + \beta_{i,k}^{\mathtt{D},(\kappa)}\frac{1}{2\alpha_i-\alpha_i^{(\kappa)}}\Bigr) \nonumber\\
&:=& \mathcal{W}^{(\kappa)}(\beta_{i,k}^{\mathtt{D}},\alpha_i)
\end{IEEEeqnarray}
where $\alpha_i^2$ is  linearized as $\alpha_i^{(\kappa)}(2\alpha_i-\alpha_i^{(\kappa)})$. As a result, the following inequality holds
\begin{IEEEeqnarray}{rCl}\label{eq:OP7:c3c1}
 \mathtt{a}(\beta_{i,k}^{\mathtt{D},{(\kappa)}})/\alpha_i + \mathtt{b}(\beta_{i,k}^{\mathtt{D},{(\kappa)}})\mathcal{W}^{(\kappa)}(\beta_{i,k}^{\mathtt{D}},\alpha_i) \leq \Gamma_{i,k}^{\mathtt{D}}, (i,k)\in\mathcal{D}\qquad
\end{IEEEeqnarray}
which is the convex approximation of \eqref{eq:OP7:c3b1}.

By following  steps \eqref{eq:OP7:c3}-\eqref{eq:OP7:c3c1}, we
 equivalently decompose \eqref{eq:OP7:d1} into  the following set of convex constraints:
\begin{IEEEeqnarray}{rCl}\label{eq:OP7:d3}
&&\rho_{i,\ell}^2\bar{g}_{m,i,\ell}/\beta_{i,\ell}^{\mathtt{U}}  \leq \bar{\chi}^{(\kappa)}_{m,i,\ell}(\mathbf{X}_i) + \bigl(1-\tilde{\epsilon}_{i,\ell}^{1/M}\bigr)N_{e,m}\sigma^2, \nonumber\\
&&\qquad\qquad \qquad \qquad \qquad \qquad\qquad\  \forall m\in\mathcal{M}, (i,\ell)\in\mathcal{U}, \IEEEyessubnumber\label{eq:OP7:d3a}\quad\\
&& \mathtt{a}(\beta_{i,\ell}^{\mathtt{U},{(\kappa)}})/\alpha_i + \mathtt{b}(\beta_{i,\ell}^{\mathtt{U},{(\kappa)}})\mathcal{W}^{(\kappa)}(\beta_{i,\ell}^{\mathtt{U}},\alpha_i) \leq \Gamma_{i,\ell}^{\mathtt{U}}, (i,\ell)\in\mathcal{U}\qquad\IEEEyessubnumber\label{eq:OP7:d3c}\  \end{IEEEeqnarray}
where $\beta_{i,\ell}^{\mathtt{U}}>0, \forall(i,\ell)\in\mathcal{U}$ are new variables,  $\bar{\chi}^{(\kappa)}_{m,i,\ell}(\mathbf{X}_i) \triangleq 2\bigl(\sum\nolimits_{k=1}^K\Re\{(\mathbf{w}_{i,k}^{(\kappa)})^H\bar{\mathbf{H}}_{m}\mathbf{w}_{i,k}\} +\Re\{\tr\bigl((\mathbf{V}_{i}^{(\kappa)})^H\bar{\mathbf{H}}_{m}\mathbf{V}_{i}\bigr)\}+\sum\nolimits_{j=1,j\neq\ell}^L\rho_{i,j}^{(\kappa)}\rho_{i,j}\bar{g}_{m,i,j}\bigl)  - \bar{\chi}_{m,i,\ell}(\mathbf{X}_i^{(\kappa)})$ is the inner approximation of $\bar{\chi}_{m,i,\ell}(\mathbf{X}_i)$.

\textit{Inner Approximation of Power Constraints \eqref{eq:OP4:c} and \eqref{eq:OP4:d}:} By applying \cite[Eq. (21)]{Nguyen:TCOM:17}, the  inner convex approximations for the non-convex constraints \eqref{eq:OP4:c} and \eqref{eq:OP4:d} are given as
\begin{IEEEeqnarray}{rCl}\label{eq:PowerAppro}
 \sum\nolimits_{k=1}^K\|\mathbf{w}_{1,k}\|^2 + \|\mathbf{V}_1\|_{\mathrm{F}}^2 + \frac{1}{\alpha_2}\Bigl(\sum\nolimits_{k=1}^{K}\|\mathbf{w}_{2,k}\|^2 + \|\mathbf{V}_2\|_{\mathrm{F}}^2\Bigr) \nonumber\\
 - \frac{2}{\alpha_2^{(\kappa)}}\Bigr(\sum\nolimits_{k=1}^K\Re\bigl\{(\mathbf{w}_{1,k}^{(\kappa)})^H\mathbf{w}_{1,k}\bigr\}+ \Re\left\{\tr\bigl((\mathbf{V}_1^{(\kappa)})^H\mathbf{V}_1\bigr)\right\}\Bigl)  \nonumber\\
+ \Bigl(\sum\nolimits_{k=1}^K\|\mathbf{w}_{1,k}^{(\kappa)}\|^2+ \|\mathbf{V}_1^{(\kappa)}\|_{\mathrm{F}}^2\Bigr)\frac{\alpha_2}{(\alpha_2^{(\kappa)})^2}\leq  P_{bs}^{\max} ,   \qquad\ \IEEEyessubnumber\label{eq:OP4convex:c}\\
 \rho_{1,\ell}^2 - \frac{2\rho_{1,\ell}^{(\kappa)}}{\alpha_2^{(\kappa)}}\rho_{1,\ell} + \frac{(\rho_{1,\ell}^{(\kappa)})^2}{(\alpha_2^{(\kappa)})^2}\alpha_2 \leq P_{1,\ell}^{\max},\;\forall \ell\in\mathcal{L}. \qquad\ \IEEEyessubnumber\label{eq:OP4convex:e}
\end{IEEEeqnarray}

In summary,  the following convex program, which is an inner approximation of \eqref{eq:OP4}, is solved at the $\kappa$-th iteration:
\begin{IEEEeqnarray}{rCl}\label{eq:OP8}
\underset{\mathbf{X},\eta,\boldsymbol{\Gamma},\boldsymbol{\alpha},\boldsymbol{\beta}}{\maxi}&&\quad \eta \IEEEyessubnumber\label{eq:OP8:a}\\
\st\;&& \eqref{eq:OP1:e}, \eqref{eq:changevariables:b},  \eqref{eq:OP4:e}, \eqref{eq:poscondi}, \eqref{eq:R1ktrust}, \eqref{eq:R1kConvex},  \nonumber\\ 
  && \eqref{eq:R1qConvexappro}, \eqref{eq:OP7:c31}, \eqref{eq:OP7:c3c1}, \eqref{eq:OP7:d3}, \eqref{eq:PowerAppro},\IEEEyessubnumber\label{eq:OP8:b}\\
&& \beta_{i,k}^{\mathtt{D}} > 0,  \beta_{i,\ell}^{\mathtt{U}} > 0,\  \forall(i,k)\in\mathcal{D}, (i,\ell)\in\mathcal{U}  \IEEEyessubnumber\label{eq:OP8:c}\qquad
\end{IEEEeqnarray}
to generate the next feasible point $(\mathbf{X}^{(\kappa+1)},\boldsymbol{\alpha}^{(\kappa+1)},\boldsymbol{\beta}^{(\kappa+1)})$, where $\boldsymbol{\beta}\triangleq\{\beta_{i,k}^{\mathtt{D}},  \beta_{i,\ell}^{\mathtt{U}}\}_{i\in\mathcal{I},k\in\mathcal{K},\ell\in\mathcal{L}}$. The proposed Algorithm \ref{algo:unknownCSI} outlines the steps to solve \eqref{eq:OP4}. This algorithm yields a non-decreasing sequence of  objective values, i.e., $\eta^{(\kappa+1)} \geq \eta^{(\kappa)}$ that is provable convergent  since the convex approximations satisfy the properties listed in \cite{Marks:78}.
\begin{algorithm}[t]
\begin{algorithmic}[1]
\protect\caption{Proposed path-following  algorithm to solve   \eqref{eq:OP4} }
\label{algo:unknownCSI}
\global\long\def\algorithmicrequire{\textbf{Initialization:}}
\REQUIRE  Set $\kappa:=0$ and solve \eqref{eq:OP5feasbile} to generate an initial feasible point $(\mathbf{X}^{(0)},\boldsymbol{\alpha}^{(0)},\boldsymbol{\beta}^{(0)})$.
\REPEAT
\STATE Solve \eqref{eq:OP8} with $(\mathbf{X}^{(\kappa)},\boldsymbol{\alpha}^{(\kappa)},\boldsymbol{\beta}^{(\kappa)})$ to obtain the optimal solution ($\mathbf{X}^{\star},\eta^{\star},\boldsymbol{\Gamma}^{\star},\boldsymbol{\alpha}^{\star}, \boldsymbol{\beta}^{\star}$).
\STATE Update $\mathbf{X}^{(\kappa+1)}:=\mathbf{X}^{\star},\boldsymbol{\alpha}^{(\kappa+1)}:=\boldsymbol{\alpha}^{\star},\boldsymbol{\beta}^{(\kappa+1)}:=\boldsymbol{\beta}^{\star}$.
\STATE Set $\kappa:=\kappa+1.$
\UNTIL Convergence\\
\end{algorithmic} \end{algorithm}   

\textit{Generation of the initial points:}
Initialized by any feasible $(\mathbf{X}^{(0)},\boldsymbol{\alpha}^{(0)})$ to the convex constraints \{\eqref{eq:OP1:e}, \eqref{eq:changevariables:b},  \eqref{eq:OP4:e}, \eqref{eq:poscondi}, \eqref{eq:R1ktrust}, \eqref{eq:R1kConvex}, \eqref{eq:R1qConvexappro}, \eqref{eq:PowerAppro}\}, the following convex program
\begin{IEEEeqnarray}{rCl}\label{eq:OP5feasbile}
\underset{\mathbf{X},\eta,\boldsymbol{\Gamma},\boldsymbol{\alpha}}{\maxi}&&\;\{\eta - \bar{\eta}_{\min}\} \IEEEyessubnumber\label{eq:OP5feasbile:a}\\
\st\;&& \eqref{eq:OP1:e}, \eqref{eq:changevariables:b},  \eqref{eq:OP4:e}, \eqref{eq:poscondi}, \eqref{eq:R1ktrust}, \eqref{eq:R1kConvex}, \eqref{eq:R1qConvexappro}, \eqref{eq:PowerAppro}\IEEEyessubnumber\label{eq:OP5feasbile:b}\quad\
\end{IEEEeqnarray}
without imposing Eves' constraints, is successively solved until reaching: $\{\eta - \bar{\eta}_{\min}\} \geq 0$. Herein, $\bar{\eta}_{\min} > 0$ is a given value to further improve the convergence speed of solving \eqref{eq:OP4}. The initial feasible $\boldsymbol{\beta}^{(0)}$ is then found by setting the inequalities in \eqref{eq:OP7:c3} and \eqref{eq:OP7:d3} to equalities.

\vspace{-0.3cm}
\section{Numerical Results}\label{NumericalResults}
 
\begin{table}[t]
\caption{Simulation Parameters}
	\label{parameter}
	\centering
	{\setlength{\tabcolsep}{0.2em}
\setlength{\extrarowheight}{0.1em}
		\begin{tabular}{l|l}
		\hline
				Parameter & Value \\
		\hline\hline
		    Carrier center  frequency/ System bandwidth                            & 2 GHz/ 10 MHz \\
				Distance between the FD-BS and  nearest user  & $\geq$ 10 m\\
				Noise power spectral density at the receivers & -174 dBm/Hz \\
				Path loss model for LOS,  $\mathrm{PL}_{\mathtt{LOS}}$ & 103.8 + 20.9$\log_{10}(d)$ dB\\
				Path loss model for NLOS,  $\mathrm{PL}_{\mathtt{NLOS}}$& 145.4 + 37.5$\log_{10}(d)$ dB\\
				Power budget at the FD-BS, $P_{bs}^{\max}$     & 26 dBm\\
				Power budget at  UL users, $P_{i,\ell}^{\max} = P_{\mathtt{U}}^{\max}$     & 23 dBm \\
			  FD residual SI, $\sigma_{\mathtt{SI}}$ & -75 dBm\\
				Number of antennas at  the FD-BS, $N_{\mathtt{t}} = N_{\mathtt{r}}$ & 5\\
		\hline		   				
		\end{tabular}}
\end{table}

 A small cell topology with 4 DL users $(K=2)$, 4 UL users $(L=2)$ and $M=2$ Eves is used in the numerical examples. The radius of the small cell is set to  100 m with inner circle radius of 50 m. 2 DL users and 2 UL users are randomly located in zone-1  and the remaining 2 DL users and 2 UL users are randomly
located in  zone-2. An Eve with $N_{e,m} = 2$ antennas is randomly placed in each zone.  Unless stated otherwise,  the parameters regarding the FD transmission take the values  provided in Table~\ref{parameter}, which follow  the 3GPP  specifications \cite{DUPLO}.  The entries of the fading loop channel $\mathbf{G}_{\mathtt{SI}}$ are  generated as independent and identically distributed Rician random variables with the Rician factor $K_{{\mathtt{SI}}}=5$ dB.
The CCI channel coefficient at a distance $d$  (in km) is assumed to undergo the path loss (PL) model for non-line-of-sight (NLOS) communications as $f_{i, k, \ell} = \sqrt{10^{-\mathrm{PL}_{\mathtt{NLOS}}/10}}\tilde{f}_{i, k,\ell}$, where $\mathrm{PL}_{\mathtt{NLOS}}$ is the PL in dB and $\tilde{f}_{i,\ell k}$ follows $\mathcal{CN}(0,1)$. All other channels follow the PL model for line-of-sight (LOS) communications as $\mathbf{L} = \sqrt{10^{-\mathrm{PL}_{\mathtt{LOS}}/10}}\tilde{\mathbf{L}}$, where $\mathbf{L}\in\{\mathbf{h}_{i,k}, \mathbf{g}_{i,\ell},\mathbf{H}_m,\mathbf{g}_{m,i,\ell}\}$ and the entries of $\tilde{\mathbf{L}}$ follow $\mathcal{CN}(0,1)$.
For comparison,  we consider three existing schemes:
 $(i)$ ``Conventional FD'':  under which all DL and UL users are simultaneously served during the entire communication time block (i.e., without considering fractional times and user grouping \cite{SunTWC16}); $(ii)$ ``FD non-orthogonal multiple access (FD-NOMA)'': under the same system model with ``Conventional FD,'' the DL transmission can adopt  NOMA \cite{Nguyen:JSAC:17} to further improve its  performance; $(iii)$	``HD'':   an HD system is considered where the HD-BS uses all antennas $N = N_{\mathtt{t}}+N_{\mathtt{r}}$ to serve all DL  and  UL users, albeit in two separate  communication time blocks. In such a case,  there is no SI and CCI; however,  the effective SR suffers from a reduction by  half.

\begin{figure}
    \begin{center}
				\begin{subfigure}[Average max-min SR versus $P_{bs}^{\max}$.]{
        \includegraphics[width=0.40\textwidth]{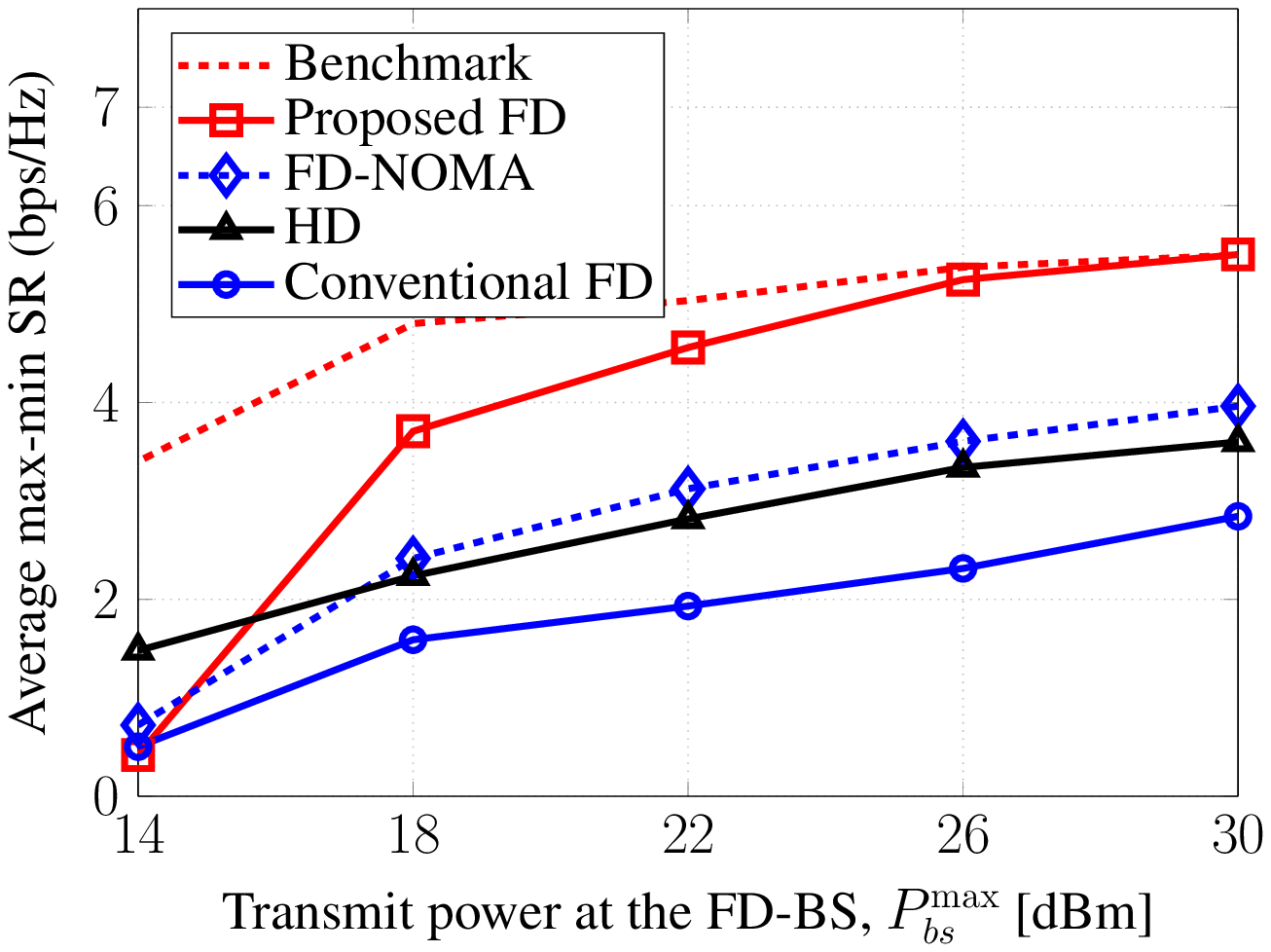}}
    		\label{fig:SCSIvsPbs:a}
				\end{subfigure}
				\begin{subfigure}[Average max-min SR of DL users versus $P_{bs}^{\max}$ for $\bar{\mathtt{R}}^{\mathtt{U}} = 2$ bps/Hz.]{
        \includegraphics[width=0.40\textwidth]{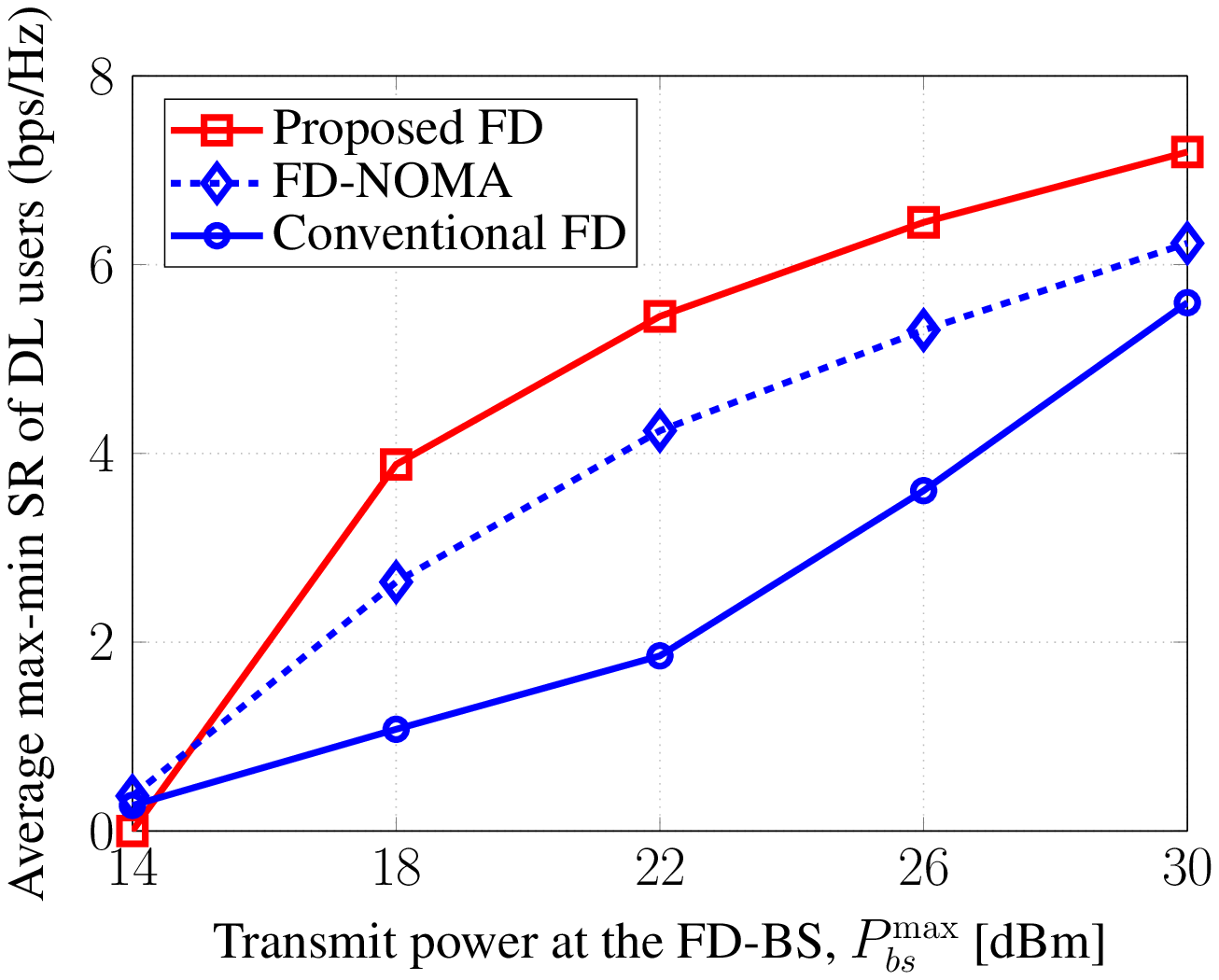}}
    		\label{fig:SCSIvsPbs:b}
				\end{subfigure}
	  \caption{(a) Average max-min SR per-user and (b) average max-min SR of DL users for $\bar{\mathtt{R}}^{\mathtt{U}} = 2$ bps/Hz, versus the transmit power at the FD-BS.}\label{fig:SCSIvsPbs}
\end{center}
\end{figure}

Fig.~\ref{fig:SCSIvsPbs}(a) depicts the average max-min SR versus the FD-BS transmit power for different resource allocation schemes.
We set $\epsilon_{i,k}  = 0.99, \forall (i,k)\in\mathcal{D}$ and $\tilde{\epsilon}_{i,\ell} = 0.99, \forall (i,\ell)\in\mathcal{U}$ to guarantee secure communications in both directions. In Fig.~\ref{fig:SCSIvsPbs}(a), we also plot a benchmark  of the proposed FD scheme, assuming  perfect CSI for the Eves.  As seen, the SRs of the proposed FD and  FD-NOMA schemes outperform the others due to the efficient proposed design and SIC, respectively.  The SR of the proposed FD scheme also approaches  that of the benchmark  when $P_{bs}^{\max}$ increases. This is because the proposed FD scheme aims to manage the network interference to improve the SR rather than concentrating the interference at Eves. At $P_{bs}^{\max} = 26$ dBm, significant gains of up to 126.8$\%$, 57.1$\%$ and 45.5$\%$ are offered by the proposed FD scheme compared to  conventional FD,  HD and  FD-NOMA, respectively. These results confirm that the proposed FD scheme is more robust and reliable in the presence of partially known Eves' CSI  compared to the others.

In a practical scenario, the DL and UL traffic demands in current generation wireless networks are typically asymmetric. Thus, we consider the following optimization problem
\begin{IEEEeqnarray}{rCl}\label{eq:OP1a}
\underset{\mathbf{X},\boldsymbol{\tau}}{\maxi}&&\;\underset{(i,k)\in\mathcal{D}}{\mini}\;\left\{R_{i,k}^{\mathtt{D}}(\mathbf{X}_i,\tau_i)\right\}, \st\  \eqref{eq:OP1:b}-\eqref{eq:OP1:f},\IEEEyessubnumber\label{eq:OP1a:c}\qquad\\
&&  R_{i,\ell}^{\mathtt{U}}(\mathbf{X}_i,\tau_i) \geq \bar{\mathtt{R}}_{i,\ell}^{\mathtt{U}},  \forall(i,\ell)\in\mathcal{U} \IEEEyessubnumber\label{eq:OP1a:b}
\end{IEEEeqnarray}
where the QoS constraints in \eqref{eq:OP1a:b} set a minimum SR requirement $\bar{\mathtt{R}}_{i,\ell}^{\mathtt{U}}$ at  UL user $(i,\ell)$.  The systematic approach in this paper is expected to be applicable for \eqref{eq:OP1a}.
The average max-min SR of the DL users versus the FD-BS transmit power  is given in Fig.~\ref{fig:SCSIvsPbs}(b) for $\bar{\mathtt{R}}_{i,\ell}^{\mathtt{U}}\equiv \bar{\mathtt{R}}^{\mathtt{U}} = 2$ bps/Hz.  The system performance of HD is not shown here due to the independence of DL and UL transmissions. As can be observed, the SRs of all schemes grow very rapidly when $P_{bs}^{\max}$  increases. The reasons behind this behavior are as follows: 1) The UL users can easily tune the power  in meeting their
 QoS requirements  to avoid strong CCI to the DL users; 2) The FD-BS will pay more attention to serve the DL users by transferring more power to them once  the UL users' QoS requirements are satisfied. Again, the proposed FD scheme achieves much better SR compared to the traditional FD schemes.

\vspace{-0.3cm}
\section{Conclusion}\label{Conclusion}
\vspace{-0.1cm}
We have addressed the problem of secure FD multiuser wireless communication. To handle the unwanted interference (SI, CCI and MUI), a simple and very efficient user grouping-based fractional time model has been proposed. We have developed a new path-following
optimization algorithm to jointly design the fractional times and power resource allocation to maximize the secrecy rate per user in both DL and UL directions. Numerical results with realistic parameters have revealed that the proposed FD scheme not only provides substantial  improvement in terms of secrecy rate over the existing schemes, but also confirms its robustness to the case when only  partial knowledge of Eves' CSI is known.

\vspace{-0.25cm}
\section*{Appendix:  Proof of Lemma \ref{lemma:1} }
Under the assumption of the independence of Eves' channels,  constraint \eqref{eq:RateEvamChange:a} can be computed as
\begin{IEEEeqnarray}{rCl}\label{eq:E1}
 \eqref{eq:RateEvamChange:a} \Leftrightarrow \Pro\bigl( C_{m,i,k}^{\mathtt{ED}}(\mathbf{X}_i,\alpha_i)\leq \Gamma_{i,k}^{\mathtt{D}}\bigr) \geq \epsilon_{i,k}^{1/M}.
\end{IEEEeqnarray}
Note that the inequality \eqref{eq:E1} holds  easier if   Eves' channels are dependent since its RHS yields a smaller value. We further rewrite \eqref{eq:E1} based on the basic property of probability as
\begin{equation}\label{eq:E2}
 \eqref{eq:E1} \Leftrightarrow \Pro\bigl( C_{m,i,k}^{\mathtt{ED}}(\mathbf{X}_i,\alpha_i)\geq\Gamma_{i,k}^{\mathtt{D}}\bigr) \leq 1-\epsilon_{i,k}^{1/M}.
\end{equation}
It requires an upper bound of the LHS of \eqref{eq:E2}, which is the outage probability for DL user ($i,k$). We make use of the  Markov inequality,  i.e., $\Pro(Y \geq y) \leq \mathbb{E}\{Y\}/y$ \cite{Billingsleybook}, to compute the LHS of \eqref{eq:E2} as
\begin{IEEEeqnarray}{rCl}\label{eq:E3}
 && \Pro\bigl(\|\mathbf{H}_{m}^H\mathbf{w}_{i,k}\|^2  + \bigl(1-e^{\alpha_i\Gamma_{i,k}^{\mathtt{D}}}\bigr)\psi'_{m,i,k}(\mathbf{X}_i) \nonumber\\
&&\qquad\qquad\qquad\qquad\qquad\;  \geq  \bigl(e^{\alpha_i\Gamma_{i,k}^{\mathtt{D}}}-1\bigr) N_{e,m}\sigma^2  \bigr)\\
&&\quad\leq\frac{\mathbb{E}\bigl\{\mathbf{w}_{i,k}^H\mathbf{H}_{m}\mathbf{H}_{m}^H\mathbf{w}_{i,k}+ \bigl(1-e^{\alpha_i\Gamma_{i,k}^{\mathtt{D}}}\bigr)\psi'_{m,i,k}(\mathbf{X}_i)\bigr\}}{\bigl(e^{\alpha_i\Gamma_{i,k}^{\mathtt{D}}}-1\bigr) N_{e,m}\sigma^2}\qquad\\
&&\quad = \frac{\mathbf{w}_{i,k}^H\bar{\mathbf{H}}_{m}\mathbf{w}_{i,k}+ \bigl(1-e^{\alpha_i\Gamma_{i,k}^{\mathtt{D}}}\bigr)\bar{\psi}_{m,i,k}(\mathbf{X}_i)}{\bigl(e^{\alpha_i\Gamma_{i,k}^{\mathtt{D}}}-1\bigr) N_{e,m}\sigma^2}\label{eq:E6}
\end{IEEEeqnarray}
where $\psi_{m,i,k}'(\mathbf{X}_i) = \psi_{m,i,k}(\mathbf{X}_i) - N_{e,m}\sigma^2$ and $\bar{\psi}_{m,i,k}(\mathbf{X}_i)$ is obtained by talking the expectation operations on each individual terms of $\psi_{m,i,k}'(\mathbf{X}_i)$.  By replacing  the LHS of \eqref{eq:E2} with \eqref{eq:E6} and  after some straightforward manipulations, we arrive at \eqref{eq:OP7:c1}. It can be shown in a similar manner that \eqref{eq:RateEvamChange:c}  is converted to \eqref{eq:OP7:d1}, and thus the proof is completed.

\vspace{-0.4cm}
\bibliographystyle{IEEEtran}
\bibliography{IEEEfull}

\begin{thebibliography}{10}
\providecommand{\url}[1]{#1}
\csname url@samestyle\endcsname
\providecommand{\newblock}{\relax}
\providecommand{\bibinfo}[2]{#2}
\providecommand{\BIBentrySTDinterwordspacing}{\spaceskip=0pt\relax}
\providecommand{\BIBentryALTinterwordstretchfactor}{4}
\providecommand{\BIBentryALTinterwordspacing}{\spaceskip=\fontdimen2\font plus
\BIBentryALTinterwordstretchfactor\fontdimen3\font minus
  \fontdimen4\font\relax}
\providecommand{\BIBforeignlanguage}[2]{{%
\expandafter\ifx\csname l@#1\endcsname\relax
\typeout{** WARNING: IEEEtran.bst: No hyphenation pattern has been}%
\typeout{** loaded for the language `#1'. Using the pattern for}%
\typeout{** the default language instead.}%
\else
\language=\csname l@#1\endcsname
\fi
#2}}
\providecommand{\BIBdecl}{\relax}
\BIBdecl

\bibitem{ZhangCM15}
Z.~Zhang, X.~Chai, K.~Long, A.~V. Vasilakos, and L.~Hanzo, ``Full duplex
  techniques for {5G} networks: {S}elf-interference cancellation, protocol
  design, and relay selection,'' \emph{IEEE Commun. Mag.}, vol.~53, no.~5, pp.
  128--137, May 2015.

\bibitem{YadavAcess17}
A.~Yadav, O.~A. Dobre, and N.~Ansari, ``Energy and traffic aware full-duplex
  communications for {5G} systems,'' \emph{IEEE Access}, vol.~5, pp.
  11\,278--11\,290, May 2017.

\bibitem{DUPLO}
\BIBentryALTinterwordspacing
``System scenarios and technical requirements for full-duplex concept,''
  \emph{DUPLO Project, Deliverable D1.1}. [Online]. Available: \url{at
  http://www.fp7-duplo.eu/index.php/deliverables.}
\BIBentrySTDinterwordspacing

\bibitem{ChenCST16}
X.~Chen, D.~W.~K. Ng, W.~Gerstacker, and H.~H. Chen, ``A survey on
  multiple-antenna techniques for physical layer security,'' \emph{IEEE Commun.
  Surveys Tutorials}, vol.~19, no.~2, pp. 1027--1053, 2nd Quarter 2017.

\bibitem{Nguyen:TIFS:16}
V.-D. Nguyen, T.~Q. Duong, O.~A. Dobre, and O.-S. Shin, ``Joint information and
  jamming beamforming for secrecy rate maximization in cognitive radio
  networks,'' \emph{IEEE Trans. Inform. Forensics $\&$ Security}, vol.~11,
  no.~11, pp. 2609--2623, Nov. 2016.

\bibitem{ZhuTSP14}
F.~Zhu, F.~Gao, M.~Yao, and H.~Zou, ``Joint information- and
  jamming-beamforming for physical layer security with full duplex base
  station,'' \emph{IEEE Trans. Signal Process.}, vol.~62, no.~24, pp.
  6391--6401, Dec. 2014.

\bibitem{ZhuTWC16}
F.~Zhu, F.~Gao, T.~Zhang, K.~Sun, and M.~Yao, ``Physical-layer security for
  full duplex communications with self-interference mitigation,'' \emph{IEEE
  Trans. Wireless Commun.}, vol.~15, no.~1, pp. 329--340, Jan. 2016.

\bibitem{SunTWC16}
Y.~Sun, D.~W.~K. Ng, J.~Zhu, and R.~Schober, ``Multi-objective optimization for
  robust power efficient and secure full-duplex wireless communication
  systems,'' \emph{IEEE Trans. Wireless Commun.}, vol.~15, no.~8, pp.
  5511--5526, Aug. 2016.

\bibitem{Nguyen:JSAC:17}
V.-D. Nguyen, H.~D. Tuan, T.~Q. Duong, H.~V. Poor, and O.-S. Shin, ``Precoder
  design for signal superposition in {MIMO-NOMA} multicell networks,''
  \emph{IEEE J. Select. Areas Commun.}, to appear, 2017.

\bibitem{Nguyen:TCOM:17}
V.-D. Nguyen, T.~Q. Duong, H.~D. Tuan, O.-S. Shin, and H.~V. Poor, ``Spectral
  and energy efficiencies in full-duplex wireless information and power
  transfer,'' \emph{IEEE Trans. Commun.}, vol.~65, no.~5, pp. 2220--2233, May
  2017.

\bibitem{Tse:book:05}
D.~Tse and P.~Viswanath, \emph{Fundamentals of Wireless Communication.}\hskip
  1em plus 0.5em minus 0.4em\relax Cambridge Univ. Press, UK, 2005.

\bibitem{Nguyen:Access:17}
V.-D. Nguyen, H.~V. Nguyen, C.~T. Nguyen, and O.-S. Shin, ``Spectral efficiency
  of full-duplex multiuser system: Beamforming design, user grouping, and time
  allocation,'' \emph{IEEE Access}, vol.~5, pp. 5785--5797, Mar. 2017.

\bibitem{AkgunTCOM17}
B.~Akgun, O.~O. Koyluoglu, and M.~Krunz, ``Exploiting full-duplex receivers for
  achieving secret communications in multiuser {MISO} networks,'' \emph{IEEE
  Trans. Commun.}, vol.~65, no.~2, pp. 956--968, Feb. 2017.

\bibitem{Marks:78}
B.~R. Marks and G.~P. Wright, ``A general inner approximation algorithm for
  nonconvex mathematical programs,'' \emph{Operations Research}, vol.~26,
  no.~4, pp. 681--683, July-Aug. 1978.

\bibitem{WES06}
A.~Wiesel, Y.~Eldar, and S.~Shamai, ``Linear precoding via conic optimization
  for fixed {MIMO} receivers,'' \emph{IEEE Trans. Signal Process.}, vol.~54,
  no.~1, pp. 161--176, Jan. 2006.

\bibitem{Beck:JGO:10}
A.~Beck, A.~Ben-Tal, and L.~Tetruashvili, ``A sequential parametric convex
  approximation method with applications to nonconvex truss topology design
  problems,'' \emph{J. Global Optim.}, vol.~47, no.~1, pp. 29--51, May 2010.

\bibitem{Billingsleybook}
P.~Billingsley, \emph{Probability and Measure}.\hskip 1em plus 0.5em minus
  0.4em\relax 3rd ed. New York: Wiley, 1995.

\end{thebibliography}
\vspace{-0.9cm}

\end{document}